\documentclass[english]{article}
\usepackage[T1]{fontenc}
\usepackage[latin9]{inputenc}
\usepackage{color}
\usepackage{amsmath}
\usepackage{amsthm}
\usepackage{amssymb}
\usepackage{graphicx}
\PassOptionsToPackage{normalem}{ulem}
\usepackage{ulem}

\makeatletter

\providecommand{\tabularnewline}{\\}

\theoremstyle{plain}
\newtheorem{thm}{\protect\theoremname}

\makeatother

\usepackage{babel}
\providecommand{\theoremname}{Theorem}

\begin{document}
\title{Signaling Games with Costly Monitoring\thanks{I would like to thank Tommaso Denti for his comments. I am most grateful
to Klaus Ritzberger for his supervision while writing this paper.}}
\author{Reuben Bearman\thanks{Department of Economics, Royal Holloway, University of London, Egham,
Surrey, TW20 0EX, United Kingdom. Email.\ Zete031@live.rhul.ac.uk }}
\date{January 2023}
\maketitle
\begin{abstract}
\textcolor{black}{If in a signaling game the receiver expects to gain
no information by monitoring the signal of the sender, then when a
cost to monitor is implemented he will never pay that cost regardless
of his off-path beliefs. This is the argument of a recent paper by
T. Denti (2021). However, which pooling equilibrium does a receiver
anticipate to gain no information through monitoring? This paper seeks
to prove that given a sufficiently small cost to monitor any pooling
equilibrium with a non-zero index will survive close to the original
equilibrium.} \newline\textbf{JEL Classification}: C72\newline \textbf{Keywords:
}Signaling games, monitoring cost, essential equilibrium.
\end{abstract}

\section{Introduction }

\textcolor{black}{In economic research wealth is an important variable.
However, this is often not directly observable or measurable. In most
sciences, these variables are replaced with a proxy, which is a variable
that offers insight into this omitted variable. However, people have
an incentive to misrepresent their wealth, for example, to reduce
their tax burden or for prestige. How can you be sure that this signal
is representative of the truth? }

\textcolor{black}{Signalling games are predominantly focused on how
truth is preserved through a signal from an informed player to an
uninformed one. Perhaps one of the most important questions derived
from this is how a modification to a base signalling game changes
the level of information passed on through a signal. The quantity
of information that is maintained can be seen through the type of
equilibrium reached.}

Two types of equilibrium can occur in signaling games, pooling and
separating, distinguished by the information that can be inferred
from the signal. In a separating equilibrium, each of the sender's
types will choose a different message. This allows the receiver to
infer the type based on the signal received and choose a best reply
with complete certainty as to which type he is facing. Pooling equilibria
have each type of the sender choose the same signal, making it impossible
for the receiver to differentiate between the different types of the
sender. Instead, the best response of the receiver will entirely rely
on his belief about the distribution of the sender's types and result
in inefficient outcomes that under perfect information would never
occur.

To see the importance of signaling games in game theory one has to
look at its history. Non-cooperative game theory relies on the notion
of complete information, that is, that there is common knowledge about
the rules of the game and each player's preferences. This severely
limits the situations that the theory can be applied to and creates
a large number of real-world situations that are left unexplained.
However, Harsanyi (1967) showed that this theory was powerful enough
to overcome these limitations, proposing a new class of games that
model this incomplete information in the form of a complete information
game. The Harsanyi transformation modifies the game so that one or
more players can have multiple types. Each player's type has its own
unique preferences, these preferences, as well as the rules of the
game, remain common knowledge. However, each player's type is known
only by the player, introducing incomplete information about the type.
The simplest form of incomplete information is a signaling game, which
has only one player with multiple types and, therefore, player 1's
choice of action becomes a message to help inform the other player
of what type player 1 could be. 

\textcolor{black}{T. Denti (2021) proposes that a player's off-path
beliefs are not the sole factor in equilibrium selection, instead
he suggests that if the sender anticipates an action to have zero
probability of occurring then the off-path beliefs about that action
can be sidestepped. This is viewed in the context of the modification
of a small cost to monitor. Costly monitoring introduces the option
for the receiver to either pay a price to observe the message sent
by the sender or to not observe the message. This decision happens
simultaneously with the sender's selection of a message. Take a situation
where the receiver expects a pooling equilibrium regardless of his
decision to monitor, that is he expects to gain no information by
observing the sender's signal. Since he expects to gain no information
by observing the senders signal as soon as a cost to monitor, no matter
how small, is implemented the receiver will never choose to monitor
the senders signal. This is not the case for every pooling equilibrium
however, T. Denti (2021) does not provide a general criterion as to
which pooling equilibrium will survive the introduction of monitoring
costs and which will not. This paper seeks to complement T. Denti
(2021) by adding further refinement as to which pooling equilibrium
there will be an equilibrium sufficiently close to the original where
the receiver still chooses to monitor. That is for which pooling equilibrium,
equilibrium analysis would depend on off-path beliefs. This paper
argues that as long as the cost-to-monitor is sufficiently small there
exists a mixed equilibrium close to the original pooling equilibrium
if the pooling equilibrium has a non-zero index. As the sender deviates
to a separating equilibrium this creates an incentive for the receiver
to monitor. This produces a mixed strategy equilibrium where the receiver
monitors a proportion of the time and likewise the sender would also
choose to pool some proportion of the time and separate the rest of
the time. The rate at which the sender separates and the receiver
chooses to not monitor increases with cost up to some threshold cost
where it is no longer optimal for the receiver to monitor thus destroying
the pooling equilibrium.}

\textcolor{black}{The contribution of this paper is to prove that
all equilibria with a non-zero index will survive a sufficiently small
cost to monitor and therefore through the methodology in T. Denti
(2021) be selected. That is not to say that the reverse is true however
as some pooling equilibrium with a zero index also have the property
of being essential and therefore will survive a small cost to monitor.
It does follow from the findings in this paper however that all pooling
equilibria that will not be selected will have an index of zero. }

\subsection{Relation to the literature}

T. Denti (2021) seeks to look at the role of off-path beliefs in equilibrium
analysis by considering signaling games with a sufficiently small
cost to monitor. The findings of this paper are that if the receiver's
off-path belief has zero probability, those beliefs can be sidestepped.
When applied to signaling games, if player 2 expects a pooling equilibrium
(i.e. the same message is selected regardless of player 1's type),
then player 2 will assign a probability of zero to any other message.
When this game is modified by introducing costly monitoring, player
2 faces a decision to pay some cost to monitor player 1's message.
If player 2 expects a pooling equilibrium, no new information would
be gained by monitoring but he would still have to pay the cost. Therefore,
not monitoring will dominate monitoring. However, this is only the
case for pooling equilibrium where the receiver does not believe that
the sender will deviate when not monitored. The contribution of this
paper is to prove that the index of pooling equilibrium is a useful
tool for determining whether they survive the modification of a small
cost to monitor. This paper finds that any pooling equilibrium with
a non-zero index survives a sufficiently sma\textcolor{black}{ll cost
to monitor with an equilibrium close to the original. The inspiration
for this proof comes from an old argument in the literature about
the survival of the first-mover advantage. This is primarily due to
the similarity between Denti's argument and that of Bagwell (1995).}

Bagwell (1995) disputes the survival of the first-mover advantage,
originally laid out by von Stackelberg (1934), when there is an error
in the observation of player 1's action. The first-mover advantage
comes from players 1's ability to commit to a strategy in a perfect
information game. Player 2 will then take this action as given and
pick the best response against it. Since player 2's strategies and
preferences are common knowledge, player 1 will choose to commit to
the strategy where player 2's best response maximizes her outcome,
hence giving player 1 an advantage. This equilibrium will be referred
to as the \textquotedblleft Stackelberg outcome\textquotedblright .
Importantly, this strategy is often different from the equilibrium
outcome when this game is played as a simultaneous-move game. This
equilibrium outcome in the simultaneous-move game is also an equilibrium
of the sequential game, though not a sub-game perfect one. Bagwell
disputes this notion by introducing an error in the observation of
player 1's action, that is, a small percentage of the time a signal
about a different play is sent to player 2. The argument goes that
player 2 will believe that any play that he sees that deviates from
the \textquotedblleft Stackelberg outcome\textquotedblright{} will
be this error message and thus will continue to play his best response
to the latter. This allows player 1 to then deviate to a more profitable
strategy given player 2's belief. The result of this is that the only
equilibrium to survive is the equilibrium outcome of the simultaneous-move
game. 

The similarities between these two papers may not be immediately obvious,
but if one views the Stackelberg game as containing a signal instead
of complete knowledge, it becomes clearer. Both papers suggest that
the response to a lack of trust in a signaling system results in that
system becoming useless. For Bagwell as soon as there is some error
introduced to the signal about player 1's action, player 2 responds
by not trusting any signal he gets from player 1, resulting in the
equilibrium of the game reverting to that of a game without monitoring
or signals. Likewise for Denti, as soon as there is a lack of faith
in the signal caused by an expected pooling equilibrium, the signal
is abandoned reverting to an equilibrium found in the game without
monitoring or signals. 

Since the cores of these papers are similar, the inspiration for this
paper came from the counter arguments to Bagwell (1995). The argument
against Bagwell comes primarily from two papers, van Damme and Hurkens
(1994) and G\"uth, Kirchsteiger, and Ritzberger (1996). Van Damme
and Hurkens (1994) find that Bagwell's conclusions are reached by
an over-reliance on pure strategy equilibrium. They found that instead,
when the error rate is small, there is a mixed strategy equilibrium
close to the original \textquotedblleft Stackelberg outcome\textquotedblright .
More importantly, this mixed strategy equilibrium tends towards the
\textquotedblleft Stackelberg outcome\textquotedblright{} as the error
decreases. G\"uth, Kirchsteiger, and Ritzberger (1996) found that,
although the extensive forms of the Stackelberg-game and its modification
with noisy observations are not comparable, the normal forms are.
The ``Stackelberg outcome'' was then shown to have a non-zero index.
An equilibrium component with a non-zero index was then shown to always
be essential, meaning that when payoffs are perturbed an equilibrium
will exist close by the original component with a non-zero index.
They conclude that the mixed equilibrium described by van Damme and
Hurkens (1994) was instead the result of the essentiality quality
of the ``Stackelberg outcome'' given a sufficiently small perturbation.

This paper argues that the similarity between Bagwell (1995) and T.
Denti (2021) leads to a similarity in solution. By applying index
theory, such as in G\"uth, Kirchsteiger, and Ritzberger (1996), a
mixed equilibrium can be found where the choice to monitor is maintained
as well as a pooling equilibrium, given a sufficiently small cost
to monitor.

This paper as well as those that inspire its proof rely heavily on
index theory, a mathematical topic within the study of topology. The
following is a brief history of the application of index theory within
game theory. For in-depth information on index theory see either McLennan
(2018) or Brown (1993). Index theory was first applied to 2-player
games by Shapley (1974). Later this was then generalized by Ritzberger
(1994) to equilibrium components of arbitrary finite games. Gul, Pearce,
and Stacchetti (1993) studied the number of regular equilibria in
generic games by using index theory. Govindan and Wilson (1997, 2005)
applied index theory to the study of refinement concepts. Demichelis
and Germano (2000) used index theory to study the geometry of the
Nash equilibrium correspondence. Demichelis and Ritzberger (2003)
applied it to evolutionary game theory. Eraslan and McLennan (2013)
applied it to coalitional bargaining.

\medskip{}

This paper is split into four sections, with both sections 2 and 3
being split into two subsections. Section 2 is entitled Preliminaries
and subsection 2.1 introduces all necessary notation to set up generic
signaling games. Section 2.2 illustrates with an example, the beer-quiche
game from Cho and Kreps (1987). Results are found in section 3.1 and
in section 3.2 the methodology of the proof is applied to the beer-quiche
game to once more illustrate the findings of this paper. Section 4
concludes.

\section{Preliminaries}

This section is broken into two parts. The first contains all relevant
definitions and concepts as well as the set up for the model. The
second illustrates this set up with an example. 

\subsection{Definitions}

This paper looks at generic 2-player signaling games. The two players
are referred to as player 1 or the \emph{sender} and player 2 or the
\emph{receiver}, these terms being used interchangeably. For clarity
the sender will be female (she/her) and the receiver male (he/his).
A generic signaling game $G=\left(T,M,A,p,u\right)$ is specified
by five objects: A finite set of types $t\in T$ assigned to player
1 by chance, a finite set of \emph{messages} or signals $m\in M$
sent by player 1 conditional on her type, a finite set of \emph{actions}
$a\in A$ taken by player 2, a \emph{probability distribution} $p:T\rightarrow\mathbb{R}_{+}$
over types (such that $\sum_{t\in T}p\left(t\right)=1$) and a \emph{payoff
function} $u=\left(u_{1},u_{2}\right):T\times M\times A\rightarrow\mathbb{R}^{2}$.
The order of play of this game goes as follows:
\begin{enumerate}
\item Player 1 is assigned a type $t\in T$ according to the probability
vector $p$ by chance.
\item Player 1 then observes her type and sends a message $m\in M$ to player
2 conditional on her type. 
\item Player 2 observes the message, but not the type of player 1, and chooses
an action $a\in A$, which ends the game. 
\end{enumerate}
Pure strategies for the signaling game $G$ for player 1 are $s_{1}\in S_{1}\equiv M^{T}=\left\{ s:T\rightarrow M\right\} $
and for player 2 the functions $s_{2}\in S_{2}\equiv A^{M}=\left\{ s:M\rightarrow A\right\} $.
Mixed strategies for the signaling game $G$ are probability distributions
over pure strategies. For player 1 they are denoted by:
\[
\sigma_{1}\in\Delta_{1}=\left\{ \sigma:M^{T}\rightarrow\mathbb{R}_{+}\left\vert \sum_{s\in S_{1}}\sigma\left(s\right)=1\right.\right\} 
\]
and for player 2 by $\sigma_{2}\in\Delta_{2}=\left\{ \sigma:A^{M}\rightarrow\mathbb{R}_{+}\left\vert \sum_{s\in S_{2}}\sigma\left(s\right)=1\right.\right\} $.
A \emph{mixed strategy profile} is a pair $\sigma=\left(\sigma_{1},\sigma_{2}\right)\in\square\equiv\Delta_{1}\times\Delta_{2}$. 

Any play $w$ in $G$ will be an ordered triple of type, message,
and action, that is, $w=\left(t,m,a\right)\in W\equiv T\times M\times A$.
An \emph{outcome} for a signaling game is a probability distribution
$\mu:W\rightarrow\mathbb{R}_{+}$ (such that $\sum_{\left(t,m,a\right)\in W}\mu\left(t,m,a\right)=1$)
over plays that is induced by some mixed strategy profile $\sigma=\left(\sigma_{1},\sigma_{2}\right)\in\square$.
The set of these outcomes will be denoted by $\Theta$.\footnote{\,Note that $\Theta$ is a proper subspace of the set of all probability
distributions on plays, because distinct players make their decisions
independently.}

The normal form for a signaling game $G$ is given by:
\[
\Gamma=\left(S_{1}\times S_{2},u\right)
\]

The definition of generic signaling games will be inspired by a result
of Kreps and Wilson (1982). Their Theorem 2 states that finite generic
extensive form games have only finitely many equilibrium outcomes.
On the other hand, Kohlberg and Mertens (1986) showed that the set
of Nash equilibria for every finite game consists of finitely many
closed and connected components. It follows that in generic extensive
form games, outcomes must be constant across every component of equilibrium.
For, if equilibrium outcomes would vary across a component then, because
the component is connected and the mapping from strategy profiles
to outcomes is continuous, there would be a continuum of equilibrium
outcomes. Therefore, in this paper a \emph{generic signaling game}
is defined as a signaling game where equilibrium outcomes are constant
across every component of equilibria.

A \emph{signaling game with costly monitoring} (SGCM) is a modification
of the signaling game $G$: The receiver does not see the sender's
message unless he decides to pay a cost $c>0$. If he does so, he
perfectly observes the sender's message. This new game $G_{c}=\left(T,M,A,p,u,c\right)$
has an additional object $c,$ which is the decision for player 2
to pay some cost $c>0$ and monitor player 1's action. Unlike in $G$,
however, the cost to monitor ($c>0$) is subtracted from player 2's
payoffs $u_{2}$. The order of play in the SGCM is as follows: 
\begin{enumerate}
\item Player 1 is assigned a type $t\in T$ according to the probability
vector $p$ by chance.
\item Player 1 then observes her type and sends a message $m\in M$ to player
2 conditional on her type. 
\item Player 2 decides simultaneously whether he would like to monitor player
1's message for a cost $c$.
\item Player 2 then observes the message only if he has previously decided
to monitor, and then chooses an action $a\in A$, ending the game.
If player 2 has not monitored, he skips the observation and goes straight
to picking some action.
\end{enumerate}
Since the decision to monitor only affects the payoffs of player 2,
player 1's payoffs are identical in both the underlying signaling
game and the SGCM. Thus player 1's pure strategies remain the same,
$S_{1}^{c}=S_{1}=M^{T}$. However, player 2's pure strategies become
the elements of $S_{2}^{c}\equiv\left\{ 0,1\right\} \times A^{M}\times A$.
The first coordinate is the decision to monitor ($s_{21}=1$) or not
($s_{21}=0$). The second is the action conditional on the message
sent by player 1. The third coordinate is the action independent of
player 1's message taken when the receiver did not monitor. Mixed
strategies for the players, $\sigma_{1}\in\Delta_{1}$ and $\sigma_{2}\in\Delta_{2}^{c}=\left\{ \sigma:\left\{ 0,1\right\} \times A^{M}\times A\rightarrow\mathbb{R}_{+}\left\vert \sum_{s\in S_{2}^{c}}\sigma\left(s\right)=1\right.\right\} $,
and mixed strategy profiles $\sigma=\left(\sigma_{1},\sigma_{2}\right)\in\square_{c}\equiv\Delta_{1}\times\Delta_{2}^{c}$
are defined accordingly. The normal form of the SGCM $G_{c}$ is defined
as:
\[
\Gamma_{c}=\left(S_{1}\times S_{2}^{c},\left(u_{1},u_{2}^{c}\right)\right)
\]
where $u_{2}^{c}\left(s\right)=u_{2}\left(s\right)-c\cdot s_{21}$
for all $s\in S_{1}\times S_{2}^{c}\equiv\square_{c}$.

Any play becomes a quadruple $\left(t,m,s_{21},a\right)\in W_{c}\equiv T\times M\times\left\{ 0,1\right\} \times A$,
where $s_{21}$ is the decision to monitor given the cost. An outcome
in a SGCM is a probability distribution $\mu_{c}:W_{c}\rightarrow\mathbb{R}_{+}$
(such that $\sum_{w\in W_{c}}\mu\left(w\right)=1$) that is induced
by some mixed strategy profile $\sigma\in\square_{c}$. This set of
outcomes will be denoted by $\Theta_{c}$. For these outcomes to project
to those of the underlying signaling game the projection must integrate
out the decision to monitor. Therefore, every outcome of a SGCM naturally
projects to an outcome of the underlying signaling game via the function
$\pi:\Theta_{c}\rightarrow\Theta$ defined by: 
\[
\pi\left(\mu_{c}\right)\left(t,m,a\right)=\mu_{c}\left(t,m,0,a\right)+\mu_{c}\left(t,m,1,a\right)\text{ for all }\left(t,m,a\right)\in W
\]

In general, two distinct strategies of the same player are \emph{strategically
equivalent} if they induce the same payoffs for all players for all
strategy profiles among the other players. If you reduce all strategically
equivalent strategies down to a single representative for all players
you create a unique \emph{purely reduced normal form}. Naturally,
in $G_{c}$ when player 2 chooses to not monitor he reaches an information
set $h_{0}$ and chooses an action $a_{0}\in A$. Likewise the same
is true when player 2 decides to monitor, he reaches an information
set $h_{m}$, dependent on the message, and chooses some action $a_{m}\in A$.
(Note that there are as many information sets $h_{m}$ as there are
messages, but there is only one $h_{0}$). It follows that since these
information sets are only reached by the decision to monitor, one
information set is always unreached. This is $h_{m}$ when player
2 monitors and $h_{0}$ when he does not monitor. All strategies that
differ only at the unreached information sets are strategically equivalent.
All strategies that differ at $h_{0}$ when player 2 monitors are
strategically equivalent, as are all strategies that differ at any
$h_{m}$ when player 2 does not monitor. Collapsing strategically
equivalent strategies for all players to single representatives gives
the purely reduced normal form $\tilde{\Gamma}_{c}=\left(\tilde{S}_{1}\times\tilde{S}_{2}^{c},\left(u_{1},u_{2}^{c}\right)\right)$
where of course $\tilde{S}_{1}=S_{1}$. The set $\tilde{S}_{2}^{c}$
of the receiver's pure strategies can be partitioned as follows: 
\begin{eqnarray*}
\tilde{S}_{2}^{c} & = & \tilde{S}_{20}^{c}\cup\tilde{S}_{21}^{c}\text{, \ }\tilde{S}_{20}^{c}\cap\tilde{S}_{21}^{c}=\emptyset\text{,}\\
\tilde{S}_{20}^{c} & = & \left\{ s_{2}\in\tilde{S}_{2}^{c}\left\vert s_{21}=0\right.\right\} \text{, \ }\tilde{S}_{21}^{c}=\left\{ s_{2}\in\tilde{S}_{2}^{c}\left\vert s_{21}=1\right.\right\} 
\end{eqnarray*}
This partitions $\tilde{S}_{2}^{c}$ into $\tilde{S}_{20}^{c}$, the
set of strategies that do not monitor and move on to $h_{0}$, and
$\tilde{S}_{21}^{c}$, the set of all strategies that do monitor and
thus move on to some $h_{m}$. The key observation from this is that
$\tilde{S}_{21}^{c}$ is bijective to $S_{2}$, the bijection being
$\left(1,\left(a_{m}\right)_{m\in M}\right)\mapsto\left(a_{m}\right)_{m\in M}$.
Furthermore, each strategy in $\tilde{S}_{20}^{c}$ also correspond
to strategies in $S_{2}$, those being all strategies in $S_{2}$
that act independently of the message. Those are the strategies that
pick the action $a$ after \emph{every} message. These strategies
will give the exact same payoffs in $\tilde{S}_{20}^{c}$ as in $S_{2}$. 

\subsection{The Beer-Quiche Game}

An example of a generic signaling game is the Beer-Quiche game borrowed
from Cho and Kreps (1987). This game will illustrate the methodology
of this paper. The Beer-Quiche game has two players, player 1 (Anne)
and player 2 (Bob). Bob has the desire to duel Anne but does not know
whether she is strong (S) or weak (W). He only wants to duel Anne
if she is weak for he fears that otherwise he will lose. It is known
by both players that Anne is the strong type 90\% of the time and
weak 10\% of the time. Each of these types has a favorite breakfast:
beer\ (B) for the strong type and quiche\ (Q) for the weak type.
Bob then observes which breakfast Anne consumes and decides whether
to challenge Anne to a duel\ (F) or not to challenge Anne to a duel
(N).

\begin{table}[t]
\centering{}%
\begin{tabular}{|c|c|c|c|c|}
\hline 
 & BB & BQ & QB & QQ\tabularnewline
\hline 
\hline 
FF & $\begin{array}{cc}
0.9\\
 & 0.1
\end{array}$ & $\begin{array}{cc}
1\\
 & 0.1
\end{array}$ & $\begin{array}{cc}
0\\
 & 0.1
\end{array}$ & $\begin{array}{cc}
0.1\\
 & 0.1
\end{array}$\tabularnewline
\hline 
FN & $\begin{array}{cc}
2.9\\
 & 0.9
\end{array}$ & $\begin{array}{cc}
2.8\\
 & 1
\end{array}$ & $\begin{array}{cc}
0.2\\
 & 0
\end{array}$ & $\begin{array}{cc}
0.1\\
 & 0.1
\end{array}$\tabularnewline
\hline 
NF & $\begin{array}{cc}
0.9\\
 & 0.1
\end{array}$ & $\begin{array}{cc}
1.2\\
 & 0
\end{array}$ & $\begin{array}{cc}
1.8\\
 & 1
\end{array}$ & $\begin{array}{cc}
2.1\\
 & 0.9
\end{array}$\tabularnewline
\hline 
NN & $\begin{array}{cc}
2.9\\
 & 0.9
\end{array}$ & $\begin{array}{cc}
3\\
 & 0.9
\end{array}$ & $\begin{array}{cc}
2\\
 & 0.9
\end{array}$ & $\begin{array}{cc}
2.1\\
 & 0.9
\end{array}$\tabularnewline
\hline 
\end{tabular}\caption{\label{tab:BQG}Normal form of the bee\textcolor{black}{r-q}uiche
game.}
\end{table}

As in the generic signaling game previously described, the beer-quiche
game is also defined by the same five objects: types, messages, actions,
a probability distribution of types, and a payoff function. The sets
of each of these objects are: types $T=\left\{ S,W\right\} $, messages
$M=\left\{ B,Q\right\} $, actions $A=\left\{ F,N\right\} $, and
a probability distribution $p\left(S\right)=0.9$ and $p\left(W\right)=0.1$.

The order of play is as follows: 
\begin{enumerate}
\item Anne is assigned a type, strong or weak, with the probability of ``strong''
being 0.9 and ``weak'' 0.1. 
\item Anne observes her type, then chooses a breakfast, either beer or quiche.
\item Bob then observes the breakfast (message) of Anne.
\item Bob then chooses either to duel or not to duel, which ends the game.
\end{enumerate}
The payoffs of this game can be seen in the extensive form shown in
Figure \ref{fig:BQG}, with the root being in the middle. For Anne,
a payoff of 1 is granted if the favorite breakfast of her type is
chosen (beer for the strong and quiche for the weak type). Anne also
has a sense of self preservation, gaining a payoff of 2 for not dueling.
For Bob a payoff of 1 is achieved if the weak type is dueled or the
strong type is not dueled, otherwise the payoff is zero.

\begin{figure}[t]
\begin{centering}
\includegraphics{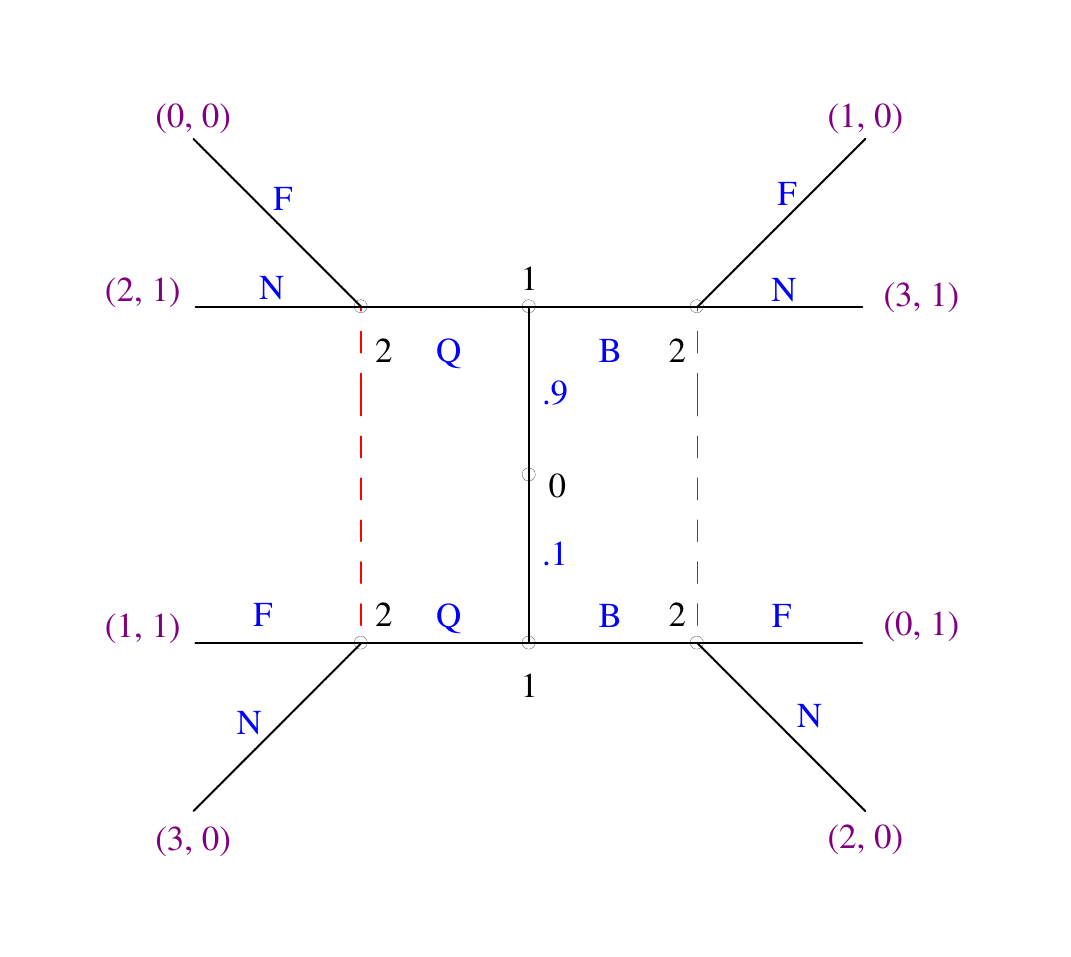}
\par\end{centering}
\caption{\label{fig:BQG}beer-quiche game in extensive form.}
\end{figure}

There are two equilibrium outcomes of this game both being pooling
equilibria. The first has Anne choosing quiche (QQ) regardless of
her type. Bob does not learn anything and responds by not dueling,
because the probability of Anne being a strong type is high. For this
equilibrium to hold, Bob must also choose to fight, when he sees Anne
drink beer, with $\Pr(F\mid B)\geq0.5$. Hence, the equilibrium component
supporting the outcome $0.9\cdot\left(S,Q,N\right)+0.1\cdot\left(W,Q,N\right)$
is a 1-dimensional line segment at the boundary of the 6-dimensional
strategy polyhedron. The result is an average payoff of 2.1 for Anne
and 0.9 for Bob. 

In the other equilibrium Anne chooses beer regardless of her type
(BB). In response Bob would choose not to fight when he sees beer.
Knowing that only the weak type has an incentive to deviate to quiche,
he would choose to fight with probability at least one-half when he
sees quiche. Hence, the equilibrium component supporting the outcome
$0.9\cdot\left(S,B,N\right)+0.1\cdot\left(W,B,N\right)$ is a 1-dimensional
line segment at the boundary of the 6-dimensional strategy polyhedron.
The result is an average payoff of 2.9 for Anne and 0.9 for Bob. Table
\ref{tab:BQG} depicts the normal form of the Beer-Quiche game, where
Anne plays columns and Bob rows. (The upper left entry is Anne's payoff
and the lower right Bob's.)

The modification of this game introduces a cost for Bob to monitor
Anne's message (breakfast). This can be imagined as any number of
scenarios. For example, Anne could consumes her breakfast in private
and Bob has to pay some cost $c>0$ to a private investigator to observe
Anne's breakfast. This modifies the extensive form as depicted in
Figure \ref{fig:BQSGCM}. (The root is again in the middle.)

\begin{table}
\centering{}%
\begin{tabular}{|c|c|c|c|c|}
\hline 
 & BB & BQ & QB & QQ\tabularnewline
\hline 
\hline 
CFFF & $\begin{array}{cc}
0.9\\
 & 0.1-c
\end{array}$ & $\begin{array}{cc}
1\\
 & 0.1-c
\end{array}$ & $\begin{array}{cc}
0\\
 & 0.1-c
\end{array}$ & $\begin{array}{cc}
0.1\\
 & 0.1-c
\end{array}$\tabularnewline
\hline 
CNFF & $\begin{array}{cc}
0.9\\
 & 0.1-c
\end{array}$ & $\begin{array}{cc}
1\\
 & 0.1-c
\end{array}$ & $\begin{array}{cc}
0\\
 & 0.1-c
\end{array}$ & $\begin{array}{cc}
0.1\\
 & 0.1-c
\end{array}$\tabularnewline
\hline 
CFFN & $\begin{array}{cc}
2.9\\
 & 0.9-c
\end{array}$ & $\begin{array}{cc}
2.8\\
 & 1-c
\end{array}$ & $\begin{array}{cc}
0.2\\
 & -c
\end{array}$ & $\begin{array}{cc}
0.1\\
 & 0.1-c
\end{array}$\tabularnewline
\hline 
CNFN & $\begin{array}{cc}
2.9\\
 & 0.9-c
\end{array}$ & $\begin{array}{cc}
1.2\\
 & -c
\end{array}$ & $\begin{array}{cc}
0.2\\
 & -c
\end{array}$ & $\begin{array}{cc}
0.1\\
 & 0.1-c
\end{array}$\tabularnewline
\hline 
CFNN & $\begin{array}{cc}
2.9\\
 & 0.9-c
\end{array}$ & $\begin{array}{cc}
3\\
 & 0.9-c
\end{array}$ & $\begin{array}{cc}
2\\
 & 0.9-c
\end{array}$ & $\begin{array}{cc}
2.1\\
 & 0.9-c
\end{array}$\tabularnewline
\hline 
CNNN & $\begin{array}{cc}
2.9\\
 & 0.9-c
\end{array}$ & $\begin{array}{cc}
3\\
 & 0.9-c
\end{array}$ & $\begin{array}{cc}
2\\
 & 0.9-c
\end{array}$ & $\begin{array}{cc}
2.1\\
 & 0.9-c
\end{array}$\tabularnewline
\hline 
CFNF & $\begin{array}{cc}
0.9\\
 & 0.1-c
\end{array}$ & $\begin{array}{cc}
1.2\\
 & -c
\end{array}$ & $\begin{array}{cc}
1.8\\
 & 1-c
\end{array}$ & $\begin{array}{cc}
2.1\\
 & 0.9-c
\end{array}$\tabularnewline
\hline 
CNNF & $\begin{array}{cc}
0.9\\
 & 0.1-c
\end{array}$ & $\begin{array}{cc}
1.2\\
 & -c
\end{array}$ & $\begin{array}{cc}
1.8\\
 & 1-c
\end{array}$ & $\begin{array}{cc}
2.1\\
 & 0.9-c
\end{array}$\tabularnewline
\hline 
0FFF & $\begin{array}{cc}
0.9\\
 & 0.1
\end{array}$ & $\begin{array}{cc}
1\\
 & 0.1
\end{array}$ & $\begin{array}{cc}
0\\
 & 0.1
\end{array}$ & $\begin{array}{cc}
0.1\\
 & 0.1
\end{array}$\tabularnewline
\hline 
0FFN & $\begin{array}{cc}
0.9\\
 & 0.1
\end{array}$ & $\begin{array}{cc}
1\\
 & 0.1
\end{array}$ & $\begin{array}{cc}
0\\
 & 0.1
\end{array}$ & $\begin{array}{cc}
0.1\\
 & 0.1
\end{array}$\tabularnewline
\hline 
0FNF & $\begin{array}{cc}
0.9\\
 & 0.1
\end{array}$ & $\begin{array}{cc}
1\\
 & 0.1
\end{array}$ & $\begin{array}{cc}
0\\
 & 0.1
\end{array}$ & $\begin{array}{cc}
0.1\\
 & 0.1
\end{array}$\tabularnewline
\hline 
0FNN & $\begin{array}{cc}
0.9\\
 & 0.1
\end{array}$ & $\begin{array}{cc}
1\\
 & 0.1
\end{array}$ & $\begin{array}{cc}
0\\
 & 0.1
\end{array}$ & $\begin{array}{cc}
0.1\\
 & 0.1
\end{array}$\tabularnewline
\hline 
0NFF & $\begin{array}{cc}
2.9\\
 & 0.9
\end{array}$ & $\begin{array}{cc}
3\\
 & 0.9
\end{array}$ & $\begin{array}{cc}
2\\
 & 0.9
\end{array}$ & $\begin{array}{cc}
2.1\\
 & 0.9
\end{array}$\tabularnewline
\hline 
0NNF & $\begin{array}{cc}
2.9\\
 & 0.9
\end{array}$ & $\begin{array}{cc}
3\\
 & 0.9
\end{array}$ & $\begin{array}{cc}
2\\
 & 0.9
\end{array}$ & $\begin{array}{cc}
2.1\\
 & 0.9
\end{array}$\tabularnewline
\hline 
0NFN & $\begin{array}{cc}
2.9\\
 & 0.9
\end{array}$ & $\begin{array}{cc}
3\\
 & 0.9
\end{array}$ & $\begin{array}{cc}
2\\
 & 0.9
\end{array}$ & $\begin{array}{cc}
2.1\\
 & 0.9
\end{array}$\tabularnewline
\hline 
0NNN & $\begin{array}{cc}
2.9\\
 & 0.9
\end{array}$ & $\begin{array}{cc}
3\\
 & 0.9
\end{array}$ & $\begin{array}{cc}
2\\
 & 0.9
\end{array}$ & $\begin{array}{cc}
2.1\\
 & 0.9
\end{array}$\tabularnewline
\hline 
\end{tabular}\caption{\label{tab:BQSGCMNF}Normal form for the beer quiche game with costly
monitoring.}
\end{table}

The order of play in this modified game is as follows: 
\begin{enumerate}
\item Anne is assigned a type $t\in T=\left\{ W,S\right\} $ according to
the probability vector $p=\left(0.9,0.1\right)$ by chance.
\item Anne then observes her type and chooses a breakfast $m\in M=\left\{ B,Q\right\} $. 
\item Bob then simultaneously decides whether or not he will pay the cost
$c>0$ to discover Anne's message (breakfast).
\item Bob then observes the message only if he has previously decided to
monitor; in either case he then chooses an action $a\in A$, ending
the game.\footnote{If Bob has not monitored, he skips the observation and goes straight
to picking some action.}
\end{enumerate}
This modification only alters the pure strategies of Bob, as he has
the extra decision to pay the cost $c>0$ and to monitor. Anne's pure
strategies however remain unchanged. Table \ref{tab:BQSGCMNF} gives
the unreduced normal form of this SGCM. (Anne is the column player
and Bob chooses a row, the upper left is Anne's payoff and the lower
right Bob's.)

\begin{figure}[t]
\centering{}\includegraphics{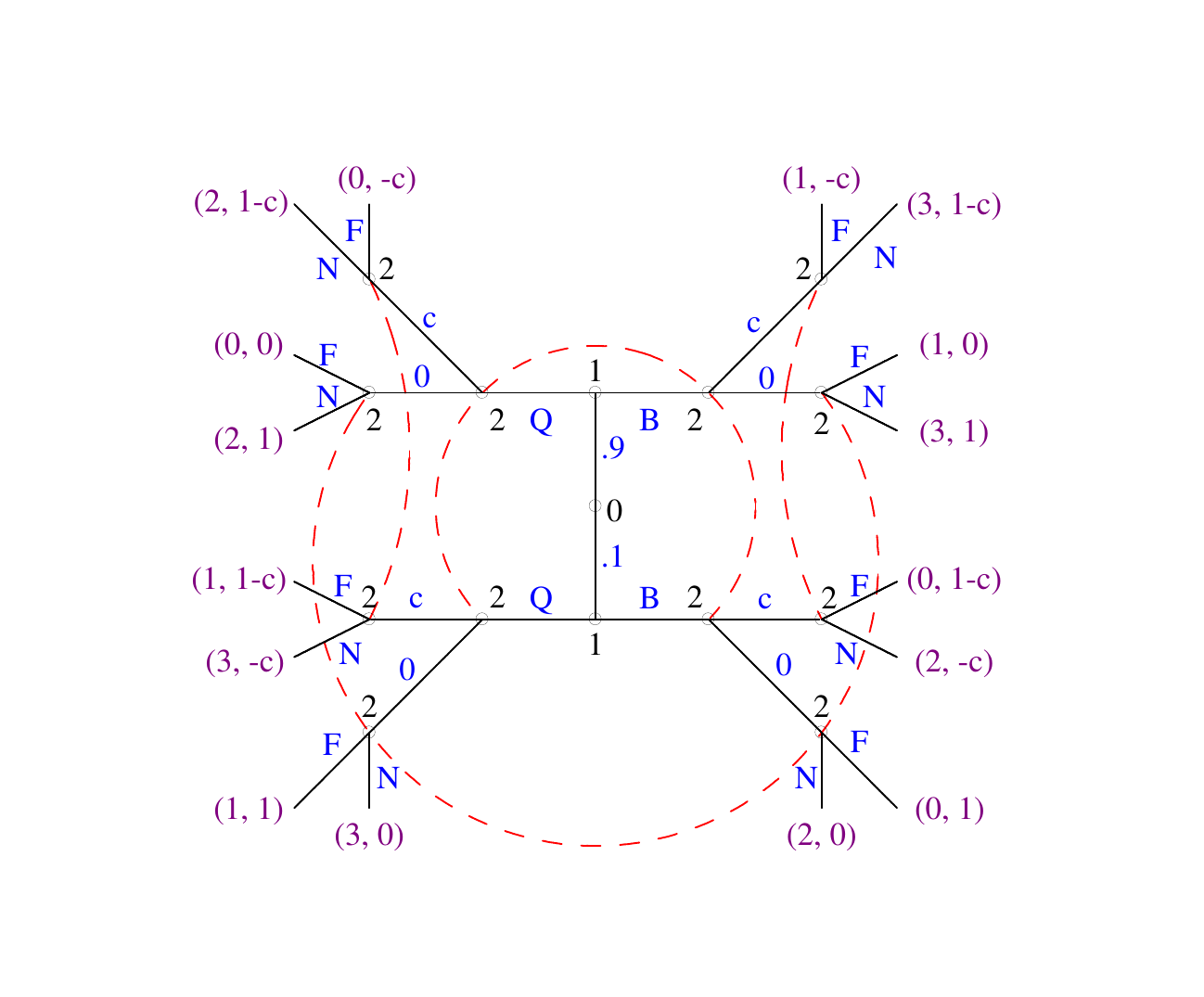}\caption{\label{fig:BQSGCM}Extensive form of the Beer-Quiche game with costly
monitoring.}
\end{figure}
 All strategies for Bob in this normal form are functions from four
information sets, $\hat{h}$, $h_{B}$, $h_{Q}$, and $h_{0}$, to
choices. As there are four information sets, the values of these functions
are a 4-dimentional vector, each coordinate representing a choice
at one information set. Here $\hat{h}$ is the information set where
Bob decides whether or not to monitor. Afterwards Bob will reach one
of the remaining 3 information sets. $h_{0}$ is the information set
reached when Bob does not monitor and thus does not observe a message.
$h_{B}$ is the information set reached when Bob observes beer. $h_{Q}$
is the information set reached when Bob observes quiche. Each strategy
can be written as a 4-dimentional vector $s_{2}\in S_{2}^{c}=\left\{ 0,1\right\} \times A^{\left\{ B,Q\right\} }\times A$.
The first coordinate is the decision to monitor $s_{21}\in\left\{ 0,1\right\} $
at $\hat{h}$. The second and third is the action $a\in A=\left\{ F,N\right\} $
at $h_{Q}$ and $h_{B}$, respectively. The last coordinate is the
action $a\in A=\left\{ F,N\right\} $ at $h_{0}$.

When Bob decides not to monitor, only the information set $h_{0}$
is reached. However in Table \ref{tab:BQSGCMNF} there are strategies
that are differentiated by decisions at the unreached information
sets $h_{B}$ and $h_{Q}$. These strategies are strategically equivalent
and can be reduced as they cause the same outcome for all players,
given a strategy of the opponent. The same can be said of strategies
where Bob decides to monitor that differ at the unreached information
set $h_{0}$. We can see, for example, that the outcomes for CFFF
are identical to those for CNFF.

\begin{table}
\centering{}%
\begin{tabular}{|c|c|c|c|c|}
\hline 
 & BB & BQ & QB & QQ\tabularnewline
\hline 
\hline 
C{*}FF & $\begin{array}{cc}
0.9\\
 & 0.1-c
\end{array}$ & $\begin{array}{cc}
1\\
 & 0.1-c
\end{array}$ & $\begin{array}{cc}
0\\
 & 1-c
\end{array}$ & $\begin{array}{cc}
0.1\\
 & 0.1-c
\end{array}$\tabularnewline
\hline 
C{*}FN & $\begin{array}{cc}
2.9\\
 & 0.9-c
\end{array}$ & $\begin{array}{cc}
2.8\\
 & 1-c
\end{array}$ & $\begin{array}{cc}
0.2\\
 & -c
\end{array}$ & $\begin{array}{cc}
0.1\\
 & 0.1-c
\end{array}$\tabularnewline
\hline 
C{*}NF & $\begin{array}{cc}
0.9\\
 & 0.1-c
\end{array}$ & $\begin{array}{cc}
1.2\\
 & -c
\end{array}$ & $\begin{array}{cc}
1.8\\
 & 1-c
\end{array}$ & $\begin{array}{cc}
2.1\\
 & 0.9-c
\end{array}$\tabularnewline
\hline 
C{*}NN & $\begin{array}{cc}
2.9\\
 & 0.9-c
\end{array}$ & $\begin{array}{cc}
3\\
 & 0.9-c
\end{array}$ & $\begin{array}{cc}
2\\
 & 0.9-c
\end{array}$ & $\begin{array}{cc}
2.1\\
 & 0.9-c
\end{array}$\tabularnewline
\hline 
0F{*}{*} & $\begin{array}{cc}
0.9\\
 & 0.1
\end{array}$ & $\begin{array}{cc}
1\\
 & 0.1
\end{array}$ & $\begin{array}{cc}
0\\
 & 0.1
\end{array}$ & $\begin{array}{cc}
0.1\\
 & 0.1
\end{array}$\tabularnewline
\hline 
0N{*}{*} & $\begin{array}{cc}
2.9\\
 & 0.9
\end{array}$ & $\begin{array}{cc}
3\\
 & 0.9
\end{array}$ & $\begin{array}{cc}
2\\
 & 0.9
\end{array}$ & $\begin{array}{cc}
2.1\\
 & 0.9
\end{array}$\tabularnewline
\hline 
\end{tabular}\caption{\label{Reduced normal form for beer-quiche game}The reduced normal
form for the beer-quiche game with costly monitoring.}
\end{table}

This creates Table \ref{Reduced normal form for beer-quiche game},
the reduced normal form of the beer-quiche game with costly monitoring.
In the reduced normal form all strategically equivalent strategies
have been collapsed into single representatives. 

Note that if we partition Bobs strategies by his decision to monitor,
those strategies where he monitors are bijective to strategies in
the underlying game and those strategies where Bob does not monitor
are identical to strategies in the underlying game under which Bob
choose the same action regardless of the breakfast he observes. 

Now consider this game when the cost to monitor is zero. Strategies
where Bob monitors are now identical to the underlying beer-quiche
game. The pure strategies, where Bob does not monitor, are also now
duplicates of the strategies where Bob chooses the same action regardless
of message (breakfast) that he sees. This means that this is no longer
the reduced normal form, as these strategies are strategically equivalent.

\section{Equilibria }

This section is broken up into two parts, the main result and an application
of the result to the example.

\subsection{Results }

This section will prove that for every non-zero index component $C$
in a generic signaling game $G$ there exists in the modified game
with costly monitoring a non-zero index component $C_{1}$ such that,
so long as the monitoring cost is sufficiently small, the equilibrium
outcome of the latter component will be found within a small neighborhood
of the equilibrium outcome of the former.

This result is achieved by first looking at the reduced normal form
of the SGCM evaluated at a cost equal to zero. This reduced normal
form is also one of the possible normal forms of the underlying signaling
game $G$, though with duplicate strategies. By Proposition 6.8(a)
of Ritzberger (2002, p.\ 324) any non-zero index equilibrium component
in a game with duplicate strategies will contain a component with
non-zero index for the game without duplicate strategies. Therefore,
a non-zero index equilibrium component of the SGCM evaluated at a
cost equal to zero contains a non-zero index component of the underlying
signaling game. Given that the underlying signaling game is generic,
it can be deduced from Theorem 2 of Kreps and Wilson (1982) that equilibrium
outcomes must be constant across every component. This implies that
the equilibrium outcome of the non-zero index component in the signaling
game with zero monitoring cost must project to an equilibrium outcome
of the original signaling game. Given that non-zero index components
are essential (Wu Wen-Ts\"un and Jiang Jia-He, 1962), when there
is a sufficiently small monitoring cost, the projected equilibrium
outcome of the signaling game with costly monitoring is close to the
equilibrium outcome associated with the non-zero index components
in the game without costly monitoring.
\begin{thm}
Let $G=\left(T,M,A,p,u\right)$ be a generic signaling game and $\left\{ G_{c}\right\} _{c>0}$
the associated family of signaling game with costly monitoring, that
is, $G_{c}=\left(T,M,A,p,u,c\right)$ with $c>0$. Then there exists
an equilibrium outcome $\mu^{*}$ for $G$ such that for every $\varepsilon>0$
there is a cost $c_{\varepsilon}>0$ with the property that, whenever
$0<c<c_{\varepsilon}$, there is an equilibrium outcome $\mu^{c}$
for $G_{c}$ satisfying $\left\Vert \pi\left(\mu^{c}\right)-\mu^{*}\right\Vert <\varepsilon$.
\end{thm}
\begin{proof}
Consider now $\tilde{\Gamma}_{0}$, that is, the reduced normal form
of the SGCM but evaluated at $c=0$. This is not a reduced normal
form anymore, because the strategies in $\tilde{S}_{20}^{0}$ are
copies of the strategies in $\tilde{S}_{21}^{0}$ that, after observing
the sender's message, still choose the same action irrespective of
the message received. Yet, $\tilde{\Gamma}_{0}$ is still a normal
form game, and in fact one in which the strategies in $\tilde{S}_{21}^{0}$
are not only bijective to $S_{2}$ but also give exactly the same
payoffs to both players as in $\Gamma$. That is, $\tilde{\Gamma}_{0}$
\emph{embeds} $\Gamma$ and $\tilde{\Gamma}_{0}$ has the same equilibria
as $\Gamma$, modulo duplicate strategies.

This implies that $\tilde{\Gamma}_{0}$ is a normal form associated
with the extensive form game $G$ and with the SGCM $G_{0}$, evaluated
at $c=0$. By Proposition 6.8(a) of Ritzberger (2002, p.\ 324) a
component of equilibria with non-zero index in a game with duplicate
strategies must contain a component of equilibria with non-zero index
of the modified game where the duplicated strategies have been deleted.
A non-zero index component $C_{0}$ in $\tilde{\Gamma}_{0}$ must
therefore contain a non-zero index component $C$ of the underlying
signaling game. Given that $G$ is generic, by definition equilibrium
outcomes must be constant across every component. Since $\tilde{\Gamma}_{0}$
is a normal form representation of the generic extensive form game
$G$, it can be concluded that all equilibrium outcomes must be constant
across every equilibrium component in $\tilde{\Gamma}_{0}$. Given
this, if $\mu^{*}$ is the equilibrium outcome of $C$, then the equilibrium
outcome $\mu^{0}$ of $C_{0}$ must project to $\mu^{*}$, that is,
$\pi\left(\mu^{0}\right)=\mu^{*}$.

By Theorem 4 of Ritzberger (1994) every component with non-zero index
is essential. That is, for every $\nu>0$ there is some $\eta_{\nu}>0$
such that, so long as the perturbed game is within $\eta{}_{\nu}$
of the original, there exists an equilibrium for the perturbed game
within $\nu$ from the essential component, with $\eta{}_{\nu}>0$
, $\nu>0$. Since the mapping from strategies to outcomes is continuous,
for every $\epsilon>0$ there must be a $\delta{}_{\epsilon}>0$ such
that strategies within $\delta_{\epsilon}$ map into outcomes within
$\epsilon$. Given that an increase in cost amounts to a perturbation
of payoff vectors, there exists some level of cost $c_{\epsilon}$
such that a SGCM with $0<c<c_{\epsilon}$ has a payoff vector within
$\eta_{\nu}$ from $\tilde{\Gamma}_{0}$. Therefore in $\tilde{\Gamma}_{0}$,
for $\epsilon>0$ choose $\delta{}_{\epsilon}>0$ such that $\left\Vert \sigma-\sigma'\right\Vert <\delta_{\epsilon}$
implies $\left\Vert \mu\left(\sigma\right)-\mu\left(\sigma'\right)\right\Vert <\epsilon$.
Now if we take a SGCM within $\eta_{\delta_{\epsilon}}>0$ of $\tilde{\Gamma}_{0}$,
then there will be an equilibrium $\sigma^{c}$ for $\tilde{\Gamma}_{c}$
such that $\min_{\sigma\in C_{0}}\left\Vert \sigma^{c}-\sigma\right\Vert <\delta_{\epsilon}$.
Hence, we can deduce that since $\sigma^{c}$ is within $\delta{}_{\epsilon}$
from $C_{0}$, the equilibrium $\sigma^{c}$ must induces an outcome
$\mu^{c}$ within $\epsilon$ from $\mu^{0}$. If we integrate out
the decision to monitor using the projection $\pi$, it follows that
$\left\Vert \pi\left(\mu^{c}\right)-\pi\left(\mu^{0}\right)\right\Vert =\left\Vert \pi\left(\mu^{c}\right)-\mu^{*}\right\Vert <\epsilon$,
as desired.
\end{proof}
This result shows that the implementation of a positive cost to monitor,
so long as these costs are sufficiently small, will result in an equilibrium
outcome close to that of the underlying game. However, it may be possible
to say more when evaluating particular classes of signaling games. 

\subsection{Beer-Quiche Game with Costly Monitoring }

In this subsection the methodology of the proof will be applied to
the Beer-Quiche example from section 2.2. The purpose is to illustrate
the technique behind the result. 

Recall the reduced normal form when the Beer-Quiche game with costly
monitoring is evaluated at $c=0$. The strategies where Bob monitors
also represent the normal form of the original game (see Table \ref{tab:BQSGCMNF}).
The strategies where Bob does not monitor are also duplicates of strategies
in the original game where Bob chooses the same action regardless
of which breakfast Anne chooses. 

Recall from Section 2.2 that Cho and Kreps identified two equilibrium
components of the Beer-Quiche game, those being Beer Beer (BB) and
Quiche Quiche (QQ). Ritzberger (1994) calculated the index for both
these equilibrium components as being +1 and 0, respectively. Therefore,
as the reduced normal form of the modified game at $c=0$ consists
of the normal form of the underlying game with duplicate strategies,
BB will also be found in the reduced normal form. This equilibrium
will be the same as in the underlying game with Anne choosing Beer
Beer and Bob responding by not fighting when he sees Beer and fighting
50\% of the time when he sees quiche. This equilibrium will have the
exact same index. 

When Anne is strong, she will have an incentive to deviate to strategies
that choose beer. Therefore, she chooses either BB or BQ. Given this,
C{*}FF and C{*}NF are always strictly dominated by C{*}FN. So long
as $0<c<8/10$ the strategy 0F{*}{*} is strictly dominated by 0N{*}{*}.
Given that there is uncertainty as to the choice of the weak type
to choose either beer or quiche, then C{*}NN is weakly dominated by
C{*}FN. This can only be a best reply if there is certainty that Anne
will only play BB. However, if this is the case for any $c>0$, both
C{*}FN and C{*}NN will be strictly dominated by 0N{*}{*}. Therefore,
C{*}NN cannot be used in any equilibrium and can be discarded. In
the resulting $2\times2$ game depicted in Table \ref{matching pennies  BQCM}
Bob has only the strategies C{*}FN and 0N{*}{*}. And, if $0<c<1/10$,
this $2\times2$ game is of the Matching Pennies variety and has a
unique equilibrium where both players randomize.

\begin{table}
\centering{}%
\begin{tabular}{|c|c|c|}
\hline 
 & $y$ & $1-y$\tabularnewline
\hline 
\hline 
$x$ & $\begin{array}{cc}
2.9\\
 & 0.9-c
\end{array}$ & $\begin{array}{cc}
2.8\\
 & 1-c
\end{array}$\tabularnewline
\hline 
$1-x$ & $\begin{array}{cc}
2.9\\
 & 0.9
\end{array}$ & $\begin{array}{cc}
3\\
 & 0.9
\end{array}$\tabularnewline
\hline 
\end{tabular}\caption{\label{matching pennies  BQCM}The resulting matching pennies game
of the beer quiche game with costly monitoring.}
\end{table}

At a mixed equilibrium both players must be indifferent between the
pure strategies in the support. Thus, to compute the mixing probabilities,
consider the indifference conditions. Those are for Anne:
\begin{align*}
u_{1}\left(x,1\right) & =2.9\textrm{ and }u_{1}\left(x,0\right)=2.8\cdot x+3\cdot\left(1-x\right)=3-0.2\cdot x\\
u_{1}\left(x,1\right) & =u_{1}\left(x,0\right)\Leftrightarrow x=1/2
\end{align*}
Thus, to keep Anne indifferent, Bob will choose to monitor with probability
of $1/2$ and fight when he sees quiche and not when he sees beer.
He will also choose to not monitor and not fight with probability
$1/2$. For player 2 the indifference condition is:
\begin{align*}
u_{2}\left(1,y\right) & =0.9\cdot y+1-y-c=1-0.1\cdot y-c\textrm{ and }u_{2}\left(0,y\right)=0.9\\
u_{2}\left(1,y\right) & =u_{2}\left(0,y\right)\Leftrightarrow y=1-10\cdot c
\end{align*}
Hence, to keep Bob indifferent, Anne consumes beer irrespective of
her type with probability $y=1-10\cdot c$ and chooses her preferred
breakfast with probability $10\cdot c$. For this to be probabilities
takes $c\leq1/10$. As the cost $c$ goes to zero, Anne will have
beer almost surely.

For a mixed strategy equilibrium to hold each player must be indifferent
between each strategy. Therefore the expected payoffs of this equilibrium
are identical to the BB component in the original beer-quiche game.
That is an expected payoff of 2.9 for player 1 and 0.9 for player
2.

To calculate the distance between the outcomes of the original beer-quiche
game and the modified game with costly monitoring one must note the
probability of plays in each game. In the modified game the plays
$\left(S,B,c,N\right)$ and $\left(S,B,0,N\right)$ occur with probability
$0.45\cdot\left(1-10c\right)$ each. The plays $\left(W,B,c,N\right)$
and $\left(W,B,0,N\right)$ occur with probability $0.05\cdot\left(1-10c\right)$
each, the plays $\left(S,Q,c,F\right)$ and $\left(S,Q,0,N\right)$
occur with probability $4.5\cdot c$ each, and the plays $\left(W,Q,c,F\right)$
and $\left(W,Q,0,N\right)$ occur with probability $0.5\cdot c$ each.
All other plays have probability zero. Integrating out the decision
whether or not to monitor gives probability $0.9\cdot\left(1-10c\right)$
for the play $\left(S,B,N\right)$, probability $0.1\cdot\left(1-10c\right)$
for the play $\left(W,B,N\right)$, probabilities $4.5\cdot c$ for
each of the two plays $\left(S,Q,F\right)$ and $\left(S,Q,N\right)$,
and probabilities $0.5\cdot c$ for each of the two plays $\left(W,Q,F\right)$
and $\left(W,Q,N\right)$. In the original beer-quiche game the BB
equilibrium induces an outcome where the play $\left(S,B,N\right)$
has probability $0.9$, the play $\left(W,B,N\right)$ has probability
$0.1$, and all other plays have probability zero. Computing Euclidean
distance between these two 8-dimensional vectors give $c\sqrt{123}\approx11.1\cdot c$,
which converges to zero as $c$ does. 

Now take the beer-quiche game with costly monitoring when evaluated
at $c>0.1$. Given that 0N{*}{*} dominates 0F{*}{*} and any strategy
that decides to monitor is dominated by the analogous strategy that
chooses not to monitor, then the pure stratergy 0N{*}{*} must be an
equilibrium strategy for Bob. Given this, player 1 will no longer
have an incentive to pool for Bob will always choose to not monitor.
This results in the separating equilibria where Anne plays beer when
she is strong and quiche when she is weak, with Bob deciding not to
monitor and chooses not to duel. This equilibria will have an expected
payoff of 3 for Anne and 0.9 for Bob. 

\section{Conclusions }

This paper looks at pooling equilibrium in generic signaling games
when a cost to monitor is introduced. This is important because this
cost to monitor represents some barrier to access information. Given
that access to information in real world situations probably has some
associated cost, if this destroys pooling equilibrium, it would have
major consequences for the validity of pooling equilibrium in real
world economic scenarios. The argument for these pooling equilibrium
being destroyed are that if player 2 expects all types of player 1
to send the same message, then no information is gained by monitoring.
Therefore it's always beneficial for player 2 to not monitor no matter
how small the cost is. This paper shows that in generic signaling
games where a pooling equilibrium is supported in a non-zero index
component, the outcome of this component survives the introduction
of a sufficiently small but positive cost to monitor. Although this
seems counter-intuitive, this is because all components with a non-zero
index are essential. Although this paper only proves this for generic
signaling games, further research may be able to prove this result
for all signaling games.

\end{document}